\documentclass[fleqn,a4paper,12pt]{article}
\usepackage{amssymb}
\usepackage{amsmath}
\usepackage{color}
\usepackage{amsthm}
\usepackage{amssymb}
\usepackage{enumitem}
\usepackage{amsfonts}
\usepackage{amstext}
\usepackage[USenglish]{babel}

\newtheorem{alemma}{Lemma}
\newtheorem{aproposition}{Proposition}
\newtheorem{atheorem}{Theorem}
\theoremstyle{definition}
\newtheorem{adefinition}{Definition}
\newtheorem{aexample}{Example}
\theoremstyle{remark}
\newtheorem{aremark}{Remark}
\definecolor{dukeblue}{rgb}{0.0, 0.0, 0.61}
\begin{document}

\date{April 17, 2023}

\author{Ata Atay\thanks{Department of Mathematical Economics, Finance and Actuarial Sciences, University of Barcelona, and Barcelona Economic Analysis Team (BEAT), Spain. E-mail: {aatay@ub.edu}} \and Sylvain Funck\thanks{CORE/LIDAM and CEREC, UCLouvain, Belgium. E-mail: {sylvain.funck@gmail.com}}  \and Ana Mauleon\thanks{CORE/LIDAM and CEREC, UCLouvain, Belgium. E-mail: {ana.c.mauleon@gmail.com}} \and Vincent Vannetelbosch\thanks{CORE/LIDAM, UCLouvain, Belgium. E-mail: {vincent.vannetelbosch@uclouvain.be}}}
\title{\textbf{Matching markets with farsighted couples}}
\maketitle

\begin{abstract}
We adopt the notion of the farsighted stable set to determine which matchings are stable when agents are farsighted in matching markets with couples. We show that a singleton matching is a farsighted stable set if and only if the matching is stable. Thus, matchings that are stable with myopic agents remain stable when agents become farsighted. Examples of farsighted stable sets containing multiple non-stable matchings are provided for markets with and without stable matchings. For couples markets where the farsighted stable set does not exist, we propose the DEM farsighted stable set to predict the matchings that are stable when agents are farsighted.

\end{abstract}

\thispagestyle{empty}\newpage

\pagenumbering{arabic}

\section{Introduction}\label{sec:intro}

Gale and Shapley (1962) propose the simple two-sided matching model with two disjoint sets of agents and the problem is to match agents from one side of the market with agents from the other side.\footnote{Roth and Sotamayor (1990) provide a complete study of the literature of two-sided matchings.} One of the main applications of this model is the entry-level labor market for medical doctors in the US (see Roth, 2003). Since the 1950s, medical students are matched to intern positions at hospitals via the National Resident Matching Program that uses the deferred acceptance (DA) algorithm of Gale and Shapley (1962) and produces a stable matching (i.e., a matching such that no doctors and hospitals have incentive to deviate from their current matching) when each doctor submits a separate preference list. However, in the presence of couples\footnote{See Roth and Peranson (1999) for details on the new design of the algorithm.} that submit joint preference lists over pairs of hospitals, a stable matching need not exist (see Roth, 1984). 

Stability has been a central problem in the literature of matching with couples. A stable matching always exists in couples markets when preferences are weakly responsive, i.e., when a unilateral improvement on acceptable positions for one partner's job is beneficial for the couple as well (Klaus and Klijn, 2005; Klaus, Klijn and Nakamura, 2009). Kojima, Pathak and Roth (2013) show that in large markets with couples a stable matching exists when there are relatively few couples and preference lists are sufficiently short relative to market size.  However, if the number of couples grows in a linear rate, a stable matching does not exist with high probability (Ashlagi, Braverman and Hassidim, 2014). Nguyen and Vohra (2018) provide a solution to this problem by perturbing slightly the capacity of hospitals.\footnote{Bir\'{o} and Klijn (2013) provides a comprehensive overview of matching markets with couples.} Furthermore, under certain conditions on preferences (weakly responsiveness in Klaus and Klijn, 2007; and restricted complementarity in Tello, 2023), it is possible to reach a stable matching from any arbitrary matching by satisfying blocking coalitions.

The notion of stable matching is a myopic notion since agents do not anticipate that individual and coalitional deviations can be followed by subsequent deviations. Other solution concepts for two-sided matching problems like the von Neumann-Morgenstern (vNM) stable set (Ehlers, 2007) or the reformulation of the vNM stable set in Herings, Mauleon and Vannetelbosch (2017), called CP vNM set, are myopic notions based on the direct dominance relation and neglect the destabilizing effect of indirect dominance relations introduced by Harsanyi (1974) and Chwe (1994). Indirect dominance captures the idea that farsighted agents can anticipate the actions of other coalitions and consider the end matching that their deviations may lead to. To the best of our knowledge, this is the first paper incorporating farsightedness in the literature on matchings with couples. 

We adopt the notion of the farsighted stable set to determine which matchings are stable when agents are farsighted in couples markets. A set of matchings is a farsighted stable set if no matching inside the set is indirectly dominated by another matching in the set (internal stability) and any matching outside the set is indirectly dominated by some matching in the set (external stability). Mauleon, Vannetelbosch and Vergote (2011) characterize the farsighted stable sets for marriage markets as all singletons that contain a stable matching. Klaus, Klijn and Walzl (2011) show that, for roommate markets, a singleton matching is a farsighted stable set if and only if the matching is stable.\footnote{Farsighted stability has also been studied in the case of hedonic games (Diamantoudi and Xue, 2004), non-transferable utility games (Ray and Vohra, 2015), and matching problems including school choice problems (Atay, Mauleon and Vannetelbosch, 2022a) and priority-based matching problems (Atay, Mauleon and Vannetelbosch, 2022b).} Thus, stable matchings remain stable when agents become farsighted in marriage markets and in roommate problems. Do stable matchings survive when agents become farsighted in couples markets? 

In this paper, we first extend the characterization of indirect dominance of Mauleon, Vannetelbosch and Vergote (2011) and Klaus, Klijn and Walzl (2011) to couples markets. We then show that, for couples markets, a singleton matching is a farsighted stable set if and only if the matching is stable. Thus, the property of stability of a matching is preserved when agents become farsighted. Moreover, a farsighted stable set cannot contain exactly two matchings. However, other farsighted stable sets with multiple non-stable matchings can exist in couples markets with stable matchings. We also study farsighted stable sets in couples markets without stable matching, and we provide an example where a non-singleton farsighted stable exists and another one without a farsighted stable set. We then introduce the DEM farsighted stable set of Herings, Mauleon and Vannetelbosch (2009, 2010) that replaces the internal stability condition in the definition of the farsighted stable set by deterrence of external deviations and minimality. Contrary to the farsighted stable set, the DEM farsighted stable set always exists. 

The paper is organized as follows. In Section \ref{sec:model}, we present the model of two-sided matching with couples. In Section \ref{sec:farsighted}, we introduce the notion of farsighted stable set. In Section \ref{stablematchings}, we study couples markets with stable matchings and we characterize singletons farsighted stable sets. We also provide an example of market with stable matching with a non-singleton farsighted stable set. In Section \ref{sec:unsolvable}, we study couples markets without stable matchings and we introduce the DEM farsighted stable set as an alternative to the non-existence of the farsighted stable set. Section \ref{sec:conc} concludes.




\section{Matching with couples}\label{sec:model}

A couples market consists of a set of hospitals $H=\{h_{1},\ldots, h_{m}\}$ and a set of students by $S=\{s_1, \ldots, s_{2n}\}$ partitioned into a set of couples $C=\{ c_1,\cdots,c_n \}=\{ (s_1,s_2),\cdots, (s_{2n-1},s_{2n})\}$. Let $s$ be a generic student and $h$ be a generic hospital.

Each hospital $h\in H$ has exactly one position to fill \footnote{We will discuss in the conclusion the case of hospitals with a capacity larger than one.} and has a strict, complete and transitive preference relation $\succeq_h$ represented by a strict ordering over the set of students $S$ and the prospect of having its position unfilled, denoted by $\emptyset$. Let $P_{h}$ denote $h$'s preferences over $S \cup \{ \emptyset\}$. A student $s$ is \textit{acceptable} for hospital $h$ if $s\succ_h \emptyset$. The preferences of all hospitals are denoted by $P_H = \{P_{h}\}_{h \in H}$.

Each couple $c\in C$ has a strict, complete and transitive preference relation $\succeq_c$ over the elements in $\mathcal{H}:= [H\cup\lbrace u\rbrace \times H\cup\lbrace u\rbrace]\setminus\lbrace(h,h)|h\in H\rbrace$ where $u$ denote the prospect of being unemployed. The preferences over the possible combinations of hospitals that a couple $c$ may reach are represented by a strict ordering $P_c$ over $\mathcal{H}$. A pair of hospitals $(h,h')\in\mathcal{H}$ are \textit{acceptable} for $c$ if $(h,h')\succ_c (u,u)$.\footnote{Whenever we use the strict part $\succ$ of a preference relation, we assume that we compare different elements in $S \cup \{ \emptyset\}$ or $\mathcal{H}$.} The preferences of all couples are denoted by $P_C = \{P_{c}\}_{c \in C}$. Let $P=(P_{H},P_{C})$ denote a couples market.\footnote{By simplicity, we assume that all students are members of a couple. However, we could easily add single students into the model since a single student corresponds to a couple where one of the member finds no hospital acceptable.} Let $\mathcal{P}$ be the set of all couples markets. 

A \textit{matching} $\mu$ for a couples market $P$ is a function $\mu:S\cup H\rightarrow S\cup H\cup \lbrace u,\emptyset\rbrace$ such that:
\begin{enumerate}[label=(\textbf{\roman*})]
\item for any $s\in S$, $\mu(s)\in H\cup\lbrace u\rbrace$,
\item for any $h\in H$, $\mu(h)\in S\cup\lbrace \emptyset\rbrace$,
\item for any $c=(s,s')\in C$, $\mu(c)=(\mu(s),\mu(s'))\in \mathcal{H}$,
\item for any $s\in S$ and $h\in H$, $\mu(\mu(s))=s$ and $\mu(\mu(h))=h$.
\end{enumerate} 

The set of matchings is denoted by $\mathcal{M}$. A student $s$ and a hospital $h$ are called \textit{mates} if $\mu(s)=h$, or equivalently $\mu(h)=s$. 

A hospital $h$ is a \textit{one-sided blocking coalition} to $\mu\in\mathcal{M}$ if $h$ is matched with an unacceptable student in $\mu$; i.e. $\emptyset \succ_{h}\mu(h)$. A couple $c$ is a \textit{one-sided blocking coalition} to $\mu\in\mathcal{M}$ if $c$ is better off by unmatching one or both partners of the couple, i.e., $c=(s_1,s_2)\in C$ is a one-sided blocking coalition if either $\mu(s_1,s_2)\prec_c (u,u)$,  $\mu(s_1,s_2)\prec_c (\mu(s_1),u)$ or $\mu(s_1,s_2)\prec_c (u,\mu(s_2))$. A matching $\mu\in\mathcal{M}$ is \textit{individually rational} if there is no one-sided blocking coalition at $\mu$.

A couple $c=(s_1,s_2)$ and two hospitals $h_1,h_2$ form a \textit{two-sided blocking coalition} to $\mu\in\mathcal{M}$ if the couple $c$ strictly prefers $(h_1,h_2)$ to their mates at $\mu$ and $h_1$, $h_2$ prefer respectively $s_1$ and $s_2$ rather than their mates at $\mu$. Formally, $\lbrace h_1,h_2,c=(s_1,s_2)\rbrace$ are a two-sided blocking coalition if: (i) $(h_1,h_2)\succ_c\mu(s_1,s_2)$, (ii) $s_1\succ_{h_1}\mu(h_1)$ and (iii) $s_2\succ_{h_2}\mu(h_2)$. A hospital $h_1$ and a couple $c=(s_1,s_2)$ are a \textit{two-sided blocking coalition} to $\mu\in\mathcal{M}$ if one partner in the couple strictly prefers $h_1$ to her mate at $\mu$ and $h_1$ prefers that student to his mate at $\mu$; i.e., if the following conditions hold: (i) $(h_1,\mu(s_2))\succ_c\mu(s_1,s_2)$ and (ii) $s_1\succ_{h_1}\mu(h_1)$. 
 
A matching $\mu\in\mathcal{M}$ is \textit{stable} if it is not blocked by any one-sided or two-sided blocking coalition. An alternative way to define a stable matching is by means of the direct dominance relation. 

Given a matching $\mu\in\mathcal{M}$, a coalition $T\in H\cup C$ can enforce a matching $\mu'\in\mathcal{M}$ over $\mu$ if any match in $\mu'$ that does not exist in $\mu$ should be between players in $T$. Moreover, if a match in $\mu$ is destroyed, then one of the two players involved in that match should belong to $T$. The next definition formalizes these ideas.

\begin{adefinition}\label{Def:Enfor} 
Given a matching $\mu$, a coalition $T\in H\cup C$ can enforce a matching $\mu'$ over $\mu$ if the following conditions holds: For every $c\in C$ and $h\in H$, $(i)$ $\mu'(h)\in c$ and $\mu'(h)\neq\mu(h)$ implies $\lbrace c, h\rbrace\subseteq T$ and $(ii)$ $\mu(h)\in c$ and $\mu'(h)=\emptyset$ implies $\lbrace c,h\rbrace\cap T\neq\emptyset$.
\end{adefinition}
Using this notion, we can describe the notion of direct dominance. 
\begin{adefinition}\label{Def:Direct} 
A matching $\mu$ is directly dominated by $\mu'$, or $\mu' > \mu$, if there exists a coalition $T\in H\cup C$ that can enforce $\mu'$ over $\mu$ and such that $\mu' \succ_i  \mu$ for all $i \in T$.
\end{adefinition}
A matching $\mu\in\mathcal{M}$ is stable if it is not directly dominated by any other matching. 

Given a matching $\mu\in\mathcal{M}$ with student $s\in S$ assigned to hospital $h\in H$, so $\mu(s)=h$, the matching $\mu'$ that is identical to $\mu$, except that the match between $s$ and $h$ has been destroyed by either $s$ or $h$, is denoted $\mu'=\mu-(s,h)$. Given a matching $\mu\in\mathcal{M}$ such that $s\in S$ and $h\in H$ are not matched to one another, the matching $\mu'$ that is identical to $\mu$, except that the pair $(s,h)$ has formed at $\mu'$ and their partners at $\mu$; i.e., $\mu(s)$ and $\mu(h)$, become unmatched at $\mu'$, is denoted by $\mu'=\mu+(s,h)$.

\section{Farsighted stable sets}\label{sec:farsighted}

The concept of stable matching is a myopic notion since the players do not anticipate that individual and coalitional deviations could be followed by subsequent deviations. This concept and other concepts based on the direct dominance relation neglect the destabilizing effect of indirect dominance relations as introduced by Harsanyi (1974) and Chwe (1994). Indirect dominance captures the idea that farsighted players can anticipate the actions of other coalitions and consider the end matching that their deviations may lead to.

\begin{adefinition}\label{def:indirect_dom}
A matching $\mu$ is indirectly dominated by $\mu'$, or $\mu'\gg\mu$, if there exists a sequence of matchings $\mu_0,...,\mu_K$ (with $\mu_0=\mu$ and $\mu_K=\mu'$) and a sequence of coalitions $T_0,...,T_{K-1}\subseteq H\cup C$ such that for any $k\in {1,...,K}$, the following conditions hold : 
\begin{enumerate}
\item Coalition $T_k$ can enforce the matching $\mu_{k+1}$ over $\mu_k$, 
\item For all $i\in T_k$, $\mu'\succ_i \mu_k$. 
 \end{enumerate}
\end{adefinition}  

The indirect dominance relation is denoted by $\ll$. Obviously, if $\mu'>\mu$, then $\mu'\gg\mu$.\footnote{Mauleon, Molis, Vannetelbosch and Vergote (2014) introduce the condition of dominance invariance to guarantee the equivalence between the direct and the indirect dominance relations. Under dominance invariance the differences between a farsighted solution concept and its myopic counterpart vanish.} Based on the indirect dominance relation, the farsighted stable set (see Chwe (1994) and Mauleon, Vannetelbosch and Vergote (2011)) is defined as the set of matchings satisfying internal and external stability conditions.

\begin{adefinition}\label{Def:farsightedstableset} 
Let $\langle H, C, P \rangle$ be a couples market. A set of matchings $V\subseteq\mathcal{M}$ is a farsighted stable set if it satisfies the following conditions:
\begin{enumerate}[label=(\textbf{\roman*})]
\item For every $\mu\in V$, there is no $\mu'\in V$ such that $\mu'\gg\mu$, 
\item For every $\mu\notin V$, there exists $\mu'\in V$ such that $\mu'\gg\mu$.
\end{enumerate}
\end{adefinition} 

Condition (i) in Definition \ref{Def:farsightedstableset} is the internal stability (IS) condition: no matching inside the set is indirectly dominated by a matching belonging to the set. Condition (ii) is the external stability (ES) condition establishing that any matching outside the set is indirectly dominated by some matching belonging to the set. 

\begin{aexample}[Klaus and Klijn, 2007] Consider a couples market with $H=\{h_1,h_2,h_3\}$ and $C=\{(s_1,s_2),(s_3,s_4)\}$. Hospitals and couples' preferences are as follows:

\begin{table}[!ht]
\centering
\begin{tabular}{cccc}
\multicolumn{3}{c}{Hospitals}\\
\hline
$P_{h_{1}}$ & $P_{h_{2}}$ & $P_{h_{3}}$ \\
\hline 
$s_2$ & $s_2$ & $s_2$ \\
$s_1$ & $s_3$ & $s_4$ \\
      & $s_1$ & $s_1$
\end{tabular}
~\qquad~      
    \begin{tabular}{cc}
\multicolumn{2}{c}{Students}\\
\hline
$P_{(s_1, s_2)}$& $P_{(s_3, s_4)}$ \\
\hline
 $(h_3,h_1)$ & $(h_2,h_3)$ \\ 
$(h_2,h_3)$ &   \\ 
 $(h_1,h_2)$ & 
\end{tabular} 
\end{table}

There exist four individually rational matchings: $\mu_1=\{(s_1,h_3),(s_2,h_1),(s_3,u),$ $(s_4,u)\}$, $\mu_2=\{(s_1,h_1),(s_2,h_2),(s_3,u),(s_4,u)\}$, $\mu_3=\{(s_1,u),(s_2,u),(s_3,h_2),(s_4,h_3)\}$ and $\mu_4=\{(s_1,h_2),(s_2,h_3),(s_3,u),(s_4,u)\}$. At $\mu_1$, the coalition $\{h_2,h_3,(s_3,s_4)\} $ would be better off in $\mu_3$ and then form a two-sided blocking coalition of $\mu_1$. Since the coalition $\{h_2,h_3,(s_3,s_4)\} $ can enforce $\mu_3$ from $\mu_1$, we have that $\mu_3>\mu_1$. At $\mu_3$, the coalition $\{h_1,h_2,(s_1,s_2)\} $ form a two-sided blocking coalition of $\mu_3$ and can enforce $\mu_2$ so that $\mu_2>\mu_3$.\footnote{Notice that if $h_3$ would have been able to anticipate the next deviation from $\mu_3$ to $\mu_2$, it would not have deviated from $\mu_1$ to $\mu_3$. Thus, $\mu_1\not\ll\mu_2$.} From $\mu_2$, the coalition $\{h_3,h_1,(s_1,s_2)\} $ form a two-sided blocking coalition of $\mu_2$ and can enforce $\mu_1$ so that $\mu_1>\mu_2$. Thus, we have that $\mu_1>\mu_2>\mu_3>\mu_1\not<\mu_2\not<\mu_3\not<\mu_1$.

In this example, $\mu_4$ is the unique stable matching. However, $\mu_4$ does not directly dominate the three other individually rational matchings. But, starting at $\mu_1$, looking forward to $\mu_4$, $h_3$ will destroy its match with $s_1$ reaching the matching $\mu_1-(s_1,h_3)$. After the deviation of $h_3$, the couple $(s_1,s_2)$ is matched with $(u,h_1)$, which is unacceptable for them. At $\mu_1-(s_1,h_3)$, the coalition $\{h_2,h_3,(s_1,s_2)\} $ can enforce and will be better off at $\mu_4$. Hence, the sequence of matchings $\mu^0,\mu^1,\mu^2$ (where $\mu^0=\mu_1$, $\mu^1=\mu^0-(s_1,h_3)$, and $\mu^2=\mu^1+\{h_2,h_3,(s_1,s_2)\}=\mu_4$) and the coalitions $T^0,T^1$ with $T^0=\{h_3\}$ and $T^1=\{h_2,h_3,(s_1,s_2)\}$ are such that $T^0$ can enforce $\mu^1$ over $\mu^0$ and coalition $T^1$ can enforce $\mu^2$ over $\mu^1$. Moreover, $\mu^2\succ \mu^0$ for $T^0$ and $\mu^2\succ \mu^1$ for $\{h_2,h_3,(s_1,s_2)\} $. Hence, we have that $\mu^2=\mu_4\gg \mu_1=\mu^0$. Similar arguments can be used to show that $\mu_4$ indirectly dominates $\mu_2$, $\mu_3$ as well as any other matching $\mu\neq \mu_4$. Therefore, $\lbrace\mu_4\rbrace$ is a farsighted stable set.\footnote{Klaus and Klijn (2007) used this example to show that a (myopic) path to a stable matching obtained from satisfying blocking coalitions does not always exist.}

Since direct dominance implies indirect dominance, we also have that $\mu_1\gg\mu_2\gg\mu_3\gg\mu_1$. On the contrary, starting at $\mu_3$ no agent has incentive to deviate looking forward to $\mu_1$. Indeed, all agents matched at $\mu_3$ formed a two-sided blocking coalition to $\mu_1$. So, $\mu_3\not\ll\mu_1 $. Similar arguments can be used to show that $\mu_1\not\ll\mu_2\not\ll\mu_3\not\ll\mu_1$. Thus, none of these three individually rational matchings form a singleton farsighted stable set. Since any singleton not individually rational matching does not indirectly dominate the individually rational matchings, we have that $\lbrace\mu_4\rbrace$ is the unique singleton farsighted stable set. Farsighted stable sets that contain two (or more) matchings do not exist. Notice that, in order to satisfy external stability, the set should at least contain two of the tree individually rational matchings. But then, since there exists no two individually rational matchings that do not indirectly dominate one another, internal stability will never be satisfied in any set containing two of the three individually rational matchings. Therefore, the unique stable matching of this couples market when agents are myopic is also the only farsighted stable set.
\end{aexample}

In this example, we have shown that a singleton stable matching is the unique farsighted stable set. Mauleon, Vannetelbosch and Vergote (2011) characterized farsighted stable sets for marriage markets and showed that a set of matchings is a farsighted stable set if and only if it is a singleton stable matching. Klaus, Klijn and Walzl (2011) showed that in roommate problems a singleton is a farsighted stable set if and only if the matching is stable. In the next section we study whether these results hold for couples markets. Furthermore, we investigate the existence of non-singleton farsighted stable sets in couples markets with stable matchings. 


\section{Markets with stable matchings}\label{stablematchings} 
Next lemma shows that if a matching $\mu$ indirectly dominates another matching $\mu'$, then there is no blocking coalition to $\mu$ matched at $\mu'$.
\begin{alemma}
\label{lemma:nec}
If $\mu\gg\mu'$, then $\mu'$ does not contain a blocking coalition of $\mu$. 
\end{alemma}

\begin{proof}
Suppose on the contrary that $\mu\gg\mu'$ and that there exists a two-sided blocking coalition $\lbrace h,h',c=(s,s')\rbrace$ such that $\mu'(c)=(h,h')$, $(h,h')\succ_c \mu(s,s')$, $s\succ_h \mu(h)$ and $s'\succeq_{h'}\mu(h')$. For $\mu$ to indirectly dominate $\mu'$, it must be that either $h$, or $h'$ or $c$ get unmatched along the path from $\mu'$ to $\mu$. But they all prefer $\mu'$ to $\mu$, and then they will never unmatch. Hence $\mu\not\gg\mu'$, a contradiction. The proof for a one-sided blocking coalition is similar. Since the member(s) of this one-sided coalition prefer to be unmatched that being matched at $\mu$, they will never match along the path from  $\mu'$ to $\mu$. Hence $\mu\not\gg\mu'$, a contradiction.	
\end{proof}

\begin{alemma}
\label{lemma:suff}
Consider any two matchings $\mu,\mu'\in\mathcal{M}$ such that $\mu$ is individually rational. Then $\mu \gg \mu'$ if there does not exist $h\in H$ and $c\in C$ with $\mu' (h)\in c$ such that both $c$ and $h$ strictly prefer $\mu'$ to $\mu$.
\end{alemma}

\begin{proof}
Let $B(\mu',\mu)$  be the set of hospitals and couples that strictly prefer $\mu$ to $\mu'$ and let $I(\mu',\mu)$ be the set of couples and hospitals who are indifferent between $\mu$ and $\mu'$. We prove the lemma showing that if the above condition is satisfied then $\mu \gg \mu'$ because we can construct a sequence of matchings and blocking coalitions starting at $\mu'$ and leading to $\mu$ that satisfies Definition \ref{def:indirect_dom}.\footnote{Our proof follows the proof of Lemma 1 from Mauleon, Vannetelbosch and Vergote (2011).}

Construct the following sequence of matchings from $\mu'$ to $\mu$: $\mu_0=\mu'$, $\mu_1=\mu'-B(\mu',\mu)$, and $\mu_2=\mu$. Consider the following sequence of coalitions $T_0=B(\mu',\mu)$ and $T_1=B(\mu',\mu)\cup\lbrace i\in H\cup C\setminus I(\mu',\mu)| \mu'(i)\subseteq B(\mu',\mu)\rbrace$. Notice that coalition $T_0$ can enforce $\mu_1$ over $\mu_0$ and, by definition of $B(\mu',\mu)$, we have that $\mu_2\succ \mu_0$ for $T_0$.

We show next that each hospital and student of $T_1$ (except the agents having the same matching in $\mu$ and $\mu'$)  are unmatched in $\mu_1$ and since $\mu=\mu_2$ is individually rational, they prefer $\mu_2$ to $\mu_1$. 

Given the above condition, we know that if one member of the couple $c$ was unemployed in $\mu'$ but not in $\mu$, either the couple or the hospital matched with the other member of the couple would belong to $B(\mu',\mu)$. Then, both the couple $c$ and the hospitals in $ \mu(c)$ will all be unmatched in $\mu_1$. Given the condition of Lemma \ref{lemma:suff}, it holds that if $c=(s_1,s_2)\in C$ and $ (h_1,h_2)\in \mathcal{H}$ are such that $\mu'(c)=(h_1,h_2)\neq\mu(c)$, either (i) $\mu'\succ_c\mu$ and $\mu\succeq_{h_1,h_2}\mu'$, or (ii) $\mu\succ_{c}\mu'$. In case (ii), since $c\in B(\mu',\mu)$, $h_1$ and $h_2$ will also be unmatched in $\mu_1$. In case (i), without loss of generality, $h_1\in B(\mu',\mu)$ and either $h_2\in B(\mu',\mu)$ or $\mu'(h_2)=\mu(h_2)=s_2$.\footnote{The proof would be the same if we took $h_2\in B(\mu',\mu)$ and either $h_1\in B(\mu',\mu)$ or $\mu'(h_1)=\mu(h_1)$.} If $h_2\in B(\mu',\mu)$, $c,h_1$ and $h_2$ will be unmatched in $\mu_1$ and since $\mu$ is individually rational it holds that they all strictly prefer $\mu_2$ to $\mu_1$. If $\mu'(h_2)=\mu(h_2)=s_2$, $h_1$ and $s_1$ will still be unmatched in $\mu_1$. In addition, if $\mu(s_1)\neq u$, since $\mu $ is individually rational, it holds that $(\mu(s_1),h_2)\succ_c (\emptyset,h_2)$. Otherwise, if the couple prefers that $s_1$ becomes unemployed, $c$ would be a one-sided blocking coalition of $\mu$.\footnote{If $\mu(s_1)=u$, $\mu(c)$ will already be formed in $\mu_1$.}  Thus, in all of these cases, $c$ and $h_1$ (and $h_2$ if it is not indifferent between $\mu$ and $\mu'$) strictly prefer $\mu_2$ to $\mu_1$. 

Furthermore, since all the agents who changed their matching between $\mu'$ and $\mu$ are either included in $B(\mu',\mu)$ or matched in $\mu'$ with a member of $B(\mu',\mu)$ (and then unmatched in $\mu_1$), we have that $\mu_2$ is enforceable from $\mu_1$ by $T_1$. Hence, $\mu\gg \mu'$. 
\end{proof}


\begin{aremark}
Lemma \ref{lemma:suff} is a sufficient condition for indirect dominance. Lemma \ref{lemma:nec} is a necessary condition for indirect dominance that is less restrictive. Under the condition of Lemma \ref{lemma:suff}, there does not exist any two-sided blocking coalition of $\mu$ matched in $\mu'$. Moreover, since $\mu$ is individually rational, there exists no one-sided blocking coalition of $\mu$ matched at $\mu'$. However, even if there exists no blocking coalition of $\mu$ matched at $\mu'$, there could still exist $c\in C$ and $h_1,h_2\in H$ such that $\mu'(c)=(h_1,h_2)$ with $\mu'\succ_{c,h_1} \mu$ and $\mu\succ_{h_2} \mu'$. It is for this particular case that we must extend the proof of Lemma \ref{lemma:suff} to have a condition for indirect dominance that is necessary and sufficient.
\end{aremark}

\begin{aproposition}
\label{prop:suffnec}
Consider any two matchings $\mu,\mu'\in\mathcal{M}$ such that $\mu$ is individually rational. Then $\mu \gg \mu'$ if and only if there does not exist any blocking coalition of $\mu$ matched in $\mu'$.
\end{aproposition}

\begin{proof}
The ``\emph{only if}'' part follows from Lemma \ref{lemma:nec}. 

Lemma \ref{lemma:suff} guarantees the ``\emph{if}'' part for the case where there exists no $h\in H$ and $c\in C$ with $\mu' (h)\in c$, such that both $c$ and $h$ strictly prefer $\mu'$ to $\mu$. We have then to extend the proof of Lemma \ref{lemma:suff} to the case where there exist $c\in C$, $h_1,h_2\in H$ such that $\mu'(c)=(h_1,h_2)$, $\mu'\succ_{c,h_1}\mu$ and $\mu\succ_{h_2}\mu'$.

Since $h_2 \in B(\mu',\mu)$, $h_2$ will join the first coalition $T_0$ that initiate the sequence of matchings leading from $\mu'$ to $\mu$. After the deviation of $T_0$ from $\mu'=\mu_0$ to $\mu_1$, we have $\mu_1(c)=(h_1,u)$. At $\mu_1$, it must hold that $\mu(c)\succ_c \mu_1(c)$ and/or $ \mu(h_1) \succ_{h_1} \mu_1(h_1) $. Otherwise, $\{c,h_1\}$ will form a blocking coalition of $\mu$ matched at $\mu'$, contradicting the hypothesis of non-existence of blocking coalitions of $\mu$ matched at $\mu'$. Hence, from $\mu_1$, either $c$, or $h_1$, or both, will become unmatched moving to $\mu_2$. Finally, from $\mu_2$, since all the agents who changed their matching between $\mu'$ and $\mu$ are either included in $B(\mu',\mu)$, or included in  $B(\mu_1,\mu)$, or matched in $\mu'$ with a member of $B(\mu',\mu)$ (and then unmatched in $\mu_1$ or in $\mu_2$), we have that $\mu_3 = \mu$ is enforceable from $\mu_2$ by $T_2=B(\mu',\mu)\cup B(\mu_1,\mu)\cup\lbrace i\in H\cup C\setminus I(\mu',\mu)| \mu'(i)\subseteq B(\mu',\mu)\rbrace$ and such that $\mu_3 \succ_{T_2} \mu_2$. Hence, $\mu\gg \mu'$.
\end{proof}


Proposition \ref{prop:suffnec} provides a characterization of indirect dominance. Hence, since a stable matching does not admit any blocking coalition, we can derive the following characterization of singleton farsighted stable sets.

\begin{atheorem}
\label{thm:characterization_singleton}
A singleton $V=\lbrace\mu\rbrace$ is a farsighted stable set if and only if $\mu$ is a stable matching. 
\end{atheorem}

\begin{proof}
Since $\mu$ is a stable matching, there exists no blocking coalitions of $\mu$ in any  $\mu'\in\mathcal{M}$. By Proposition \ref{prop:suffnec}, this condition is satisfied if and only if $\mu\gg\mu'$ for any $\mu'\in\mathcal{M}$, which is the definition of a singleton farsighted stable set.

Let $\lbrace\mu\rbrace$ be a farsighted stable set. By the external stability condition of Definition \ref{Def:farsightedstableset} and by Lemma \ref{lemma:nec}, there exists no blocking coalition at $\mu$ matched in any $\mu'\in \mathcal{M}$. Hence, $\mu$ is a stable matching. 
\end{proof}



Theorem \ref{thm:characterization_singleton} characterizes singleton farsighted stable sets: a singleton matching is a farsighted stable set if and only if the matching is stable. We next investigate the existence of non-singleton farsighted stable sets. In marriage markets, Mauleon, Vannetelbosch and Wergote (2011) showed that the only farsighted stable sets were the singletons consisting of a stable matching. In roommate markets, Klaus, Klijn, Walzl (2009) showed that every stable matching is a singleton farsighted stable set and that there could exist farsighted stable sets with more than one element for couples markets without stable matchings. They left the question open for markets with stable matchings. We show that farsighted stable sets with several elements can exist both for couples markets with and without stable matchings. First, using Lemma \ref{lemma:nec}, we can show that no farsighted stable set can be composed of exactly two elements.

\begin{alemma}
\label{lem:no_stableset_card_two}
There does not exist a farsighted stable set of cardinality two.
\end{alemma}

\begin{proof}
Suppose on the contrary that there exists a farsighted stable $V=\{ \mu_1,\mu_2\}$, with $\mu_1 \neq \mu_2$. Notice first that $\mu_1$ and $\mu_2$ should not be stable. Otherwise, since a stable matching indirectly dominates any other matching, the internal stability condition of Definition \ref{Def:farsightedstableset} would be violated. Second, both $\mu_1$ and $\mu_2$ should be individually rational. Otherwise, if $\mu_1$ is not individually rational, it would not indirectly dominates the matchings respecting one-sided blocking coalitions of $\mu_1$. In order to satisfy the external stability condition of Definition \ref{Def:farsightedstableset}, $\mu_2$ must then indirectly dominates the  matchings respecting one-sided blocking coalitions of $\mu_1$. But in this case $\mu_2$ would also indirectly dominates $\mu_1$ violating the internal stability condition of Definition \ref{Def:farsightedstableset}.

Since both $\mu_1$ and $\mu_2$ are unstable individually rational matchings belonging to $V=\{ \mu_1,\mu_2\}$, we can construct a matching $\mu_1'$ from $\mu_1$ enforced by the two-sided blocking coalitions of $\mu_1$ matched in $\mu_2$. Otherwise, if there does not exist any blocking coalition of $\mu_1$ matched at $\mu_2$, $\mu_1\gg\mu_2$, violating the internal stability condition of Definition \ref{Def:farsightedstableset}.\footnote{Note that $\mu_1'$ is not necessarily individually rational.} By Lemma \ref{lemma:nec} and the construction of $\mu_1'$, we have that $\mu_1\not\gg \mu_1'$ because at least one blocking coalition of $\mu_1$ is matched in $\mu_1'$. By external stability of $V$, we have that $\mu_2\gg \mu_1'$. Then, there does not exist a two-sided blocking coalition of $\mu_2$ matched in $\mu_1$. Otherwise, the two-sided blocking coalition would also be matched in $\mu_1'$. Since the two-sided blocking coalition is worse off in $\mu_2$ compared to $\mu_1'$, their members would never join any deviation leading to $\mu_2$, contradicting the fact that $\mu_2\gg\mu_1'$. Thus, by Lemma \ref{lemma:nec}, this implies that $\mu_2\gg\mu_1$, which contradicts the internal stability of $V$. 
\end{proof}

Next example shows that there might exist farsighted stable sets with more than one element for couples markets with stable matchings.

\begin{aexample}\label{ex:three-stable-set}
Consider a matching with couples problem with $H=\{h_1,h_2,h_3,h_4,h_5,h_6\}$ and $C=\{(s_1,s_2),(s_3,s_4),(s_5,s_6),(s_7,s_8),(s_9,s_{10})\}$. Hospitals and couples' preferences are as follows.

\begin{table}[!ht]
\centering
\begin{tabular}{cccccc}
\multicolumn{6}{c}{Hospitals}\\
\hline
$P_{h_{1}}$ & $P_{h_{2}}$ & $P_{h_{3}}$ & $P_{h_{4}}$ & $P_{h_{5}}$ & $P_{h_{6}}$ \\
\hline 
$s_4$ & $s_2$ & $s_5$ & $s_1$ & $s_8$ & $s_2$ \\
$s_8$ & $s_6$ & $s_7$ & $s_3$ & $s_6$ & $s_5$ \\
$s_1$ & $s_9$ & $s_3$ & $s_9$ & $s_7$ & $s_8$ \\
$s_{10}$ & & & & &
\end{tabular}
\begin{tabular}{ccccc}
\multicolumn{5}{c}{Couples}\\
\hline
$P_{(s_1, s_2)}$ & $P_{(s_3, s_4)}$ & $P_{(s_5, s_6)}$ & $P_{(s_7, s_8)}$ & $P_{(s_9, s_{10})}$ \\
\hline
$(h_1,h_2)$ & $(h_4,h_1)$ & $(h_6,h_5)$ & $(h_5,h_6)$ & $(h_2,u)$ \\ 
$(h_4,h_6)$ & $(h_3,u)$   & $(h_3,u)$   & $(h_3,h_5)$ & $(h_4,u)$ \\ 
            &             & $(u,h_2)$   & $(h_5,h_1)$ & $(u,h_1)$ \\ 
            & 			  &				&			  &			  
\end{tabular} 
\end{table}

There exists a unique stable matching $\mu=\{(s_1,h_4),(s_2,h_6),(s_3,u),(s_4,u),(s_5,h_3),\\(s_6,u),(s_7,h_5),(s_8,h_1),(s_9,h_2),(s_{10},u)\}$. Then, Theorem \ref{thm:characterization_singleton} implies that the singleton set $V=\{\mu\}$ is a farsighted stable set. Consider now the following three matchings:
\begin{equation*}
\begin{split}
	 \mu_1 &= \left\{(s_1,h_1), (s_2,h_2),(s_3,h_3),(s_4,u),(s_5,h_6),(s_6,h_5),(s_7,u),(s_8,u),(s_9,h_4),(s_{10},u)\right\}, \\
	 \mu_2 &= \{(s_1,h_4),(s_2,h_6),(s_3,u),(s_4,u),(s_5,u),(s_6,h_2),(s_7,h_3),(s_8,h_5),(s_9,u),(s_{10},h_1)\},\\
	 \mu_3 &= \{(s_1,u),(s_2,u),(s_3,h_4),(s_4,h_1),(s_5,h_3),(s_6,u),(s_7,h_5),(s_8,h_6),(s_9,h_2),(s_{10},u)\}.
\end{split}
\end{equation*}

We show now that the set $V'=\{\mu_1,\mu_2,\mu_3\}$ is another farsighted stable set.
 
The matching $\mu_2$ is blocked by the coalition $\{h_{1},h_{2},(s_{1},s_{2})\}$. Moreover, the couple $(s_{1},s_{2})$ and hospitals $h_{1}$ and $h_{2}$ are matched under the matching $\mu_{1}$. So, from Lemma \ref{lemma:nec}, $\mu_2 \not\gg \mu_1$. On the other hand, the matching $\mu_{1}$ is blocked by the coalition $\{h_{3},h_{5},(s_{7},s_{8})\}$. Moreover, the couple $(s_{7},s_{8})$ and hospitals $h_{3}$ and $h_{5}$ are matched under the matching $\mu_{2}$. Hence, from Lemma \ref{lemma:nec}, $\mu_1 \not\gg \mu_2$. 

The matching $\mu_3$ is blocked by the coalition $\{h_{5},h_{6},(s_{5},s_{6})\}$. Moreover, student $s_{5}$ and hospital $h_{6}$ and student $s_{6}$ and $h_5$ are matched under the matching $\mu_1$. So, from Lemma \ref{lemma:nec}, $\mu_3 \not\gg \mu_1$. On the other hand, the matching $\mu_{1}$ is blocked by the coalition $\{h_{4},h_{1},(s_{3},s_{4})\}$. Moreover, the couple $(s_{3},s_{4})$ and hospitals $h_{4}$ and $h_{1}$ are matched under the matching $\mu_{3}$. Hence, from Lemma \ref{lemma:nec}, $\mu_1 \not\gg \mu_3$. 

The matching $\mu_3$ is blocked by the coalition $\{h_{4},h_{6},(s_{1},s_{2})\}$. Moreover, the couple $(s_{1},s_{2})$ and hospitals $h_{4},h_{6}$ are matched under the matching $\mu_2$. So, from Lemma \ref{lemma:nec}, $\mu_3 \not\gg \mu_2$. On the other hand, the matching $\mu_{2}$ is blocked by the coalition $\{h_{3},(s_{5},s_{6})\}$. Moreover, the couple $(s_{5},s_{6})$ and hospitals $u,h_{3}$ are matched under the matching $\mu_{3}$. Hence, from Lemma \ref{lemma:nec}, $\mu_2 \not\gg \mu_3$.

Since there is no matching in the set $V'$ that indirectly dominates another matching in $V'$, the set $V'=\{\mu_1,\mu_2,\mu_3\}$ satisfies the internal stability condition of Definition \ref{Def:farsightedstableset}. 

In addition, the stable matching $\mu$ is indirectly dominated by $\mu_1$. The sequence of matchings $\mu^0=\mu$, $\mu^1=\mu^0 - \{h_5,h_2,(s_{1},s_{2}),(s_{5},s_{6})\} $, $\mu^2=\mu^1 - \{(s_{7},s_{8})\} $, $\mu^3=\mu^2 + \{h_1,h_2,h_3,h_4,h_5,h_6,(s_{1},s_{2}),(s_{3},s_{4}),(s_{5},s_{6}),(s_{9},s_{10})\} = \mu_1 $ are enforceable respectively by the coalitions $T^0=\{h_5,h_2,(s_{1},s_{2}),(s_{5},s_{6})\}$, $T^1=\{(s_{7},s_{8})\}$, $T^2=\{h_1,h_2,h_3,h_4,h_5,h_6,(s_{1},s_{2}),(s_{3},s_{4}),(s_{5},s_{6}),(s_{9},s_{10})\}$. Notice that, at each step, the members of $T^i$, $i=\{0,1,2\}$, are better off in $\mu_1$ than in $\mu^i$. Thus, $\mu_1\gg\mu$. It can easily be checked that all other matchings are indirectly dominated by some matching inside $V'$. Thus, $V'$ also satisfies the external stability condition in Definition \ref{Def:farsightedstableset}. Therefore, the set $V'=\{\mu_{1},\mu_{2},\mu_{3}\}$ is a farsighted stable set.
\end{aexample}

\section{Markets without stable matchings}
\label{sec:unsolvable}

In this section we study the farsighted stable set in couples markets without stable matchings. From Theorem \ref{thm:characterization_singleton}, we deduce that in these markets without stable matchings there is no singleton farsighted stable set. We provide two examples to show that the existence of farsighted stable sets is not guaranteed for couples markets without stable matchings. Next example provides a couples market without stable matchings but with a farsighted stable set containing several matchings.

\begin{aexample} 
Consider the couples market derived from Example \ref{ex:three-stable-set} by removing the acceptable pair of hospitals $(h_5,h_1)$ for the couple $(s_7,s_8)$ and the acceptable doctor $s_8$ for hospital $h_1$. Hospitals and couples' preferences are as follows.

\begin{table}[!ht]
\centering
\begin{tabular}{cccccc}
\multicolumn{6}{c}{Hospitals}\\
\hline
$P_{h_{1}}$ & $P_{h_{2}}$ & $P_{h_{3}}$ & $P_{h_{4}}$ & $P_{h_{5}}$ & $P_{h_{6}}$ \\
\hline 
$s_4$ & $s_2$ & $s_5$ & $s_1$ & $s_8$ & $s_2$ \\
$s_1$ & $s_6$ & $s_7$ & $s_3$ & $s_6$ & $s_5$ \\
$s_{10}$ & $s_9$ & $s_3$ & $s_9$ & $s_7$ & $s_8$
\end{tabular}
\begin{tabular}{ccccc}
\multicolumn{5}{c}{Couples}\\
\hline
$P_{(s_1, s_2)}$ & $P_{(s_3, s_4)}$ & $P_{(s_5, s_6)}$ & $P_{(s_7, s_8)}$ & $P_{(s_9, s_{10})}$ \\
\hline
$(h_1,h_2)$ & $(h_4,h_1)$ & $(h_6,h_5)$ & $(h_5,h_6)$ & $(h_2,u)$ \\ 
$(h_4,h_6)$ & $(h_3,u)$   & $(h_3,u)$   & $(h_3,h_5)$ & $(h_4,u)$ \\ 
            &             & $(u,h_2)$   & 		      & $(u,h_1)$ 	  
\end{tabular} 
\end{table}
Notice that there does not exist a stable matching. However, the set of matchings $V=\lbrace \mu_1,\mu_2,\mu_3\rbrace$ is still a farsighted stable set as we can easily check that it respects the internal and external stability conditions from Definition \ref{Def:farsightedstableset}.\footnote{This is because we have only removed some individually rational matchings from the preferences in Example 2 but the order of the preferences has not been modified.} This example shows that couples markets can exhibit a farsighted stable set containing unstable matchings, which is not the case for marriage markets. See Theorem 2 in Mauleon, Vannetelbosch and Vergote (2011).
\end{aexample}

Next example provides a couples market without a farsighted stable set.

\begin{aexample}[Roth, 2008]
Consider a couples market with $H=\{h_1,h_2\}$ and $C=\{(s_1,s_2),(s_3,s_4)\}$. Hospitals and couples' preferences are as follows.

\begin{table}[!ht]
\centering
\begin{tabular}{cc}
\multicolumn{2}{c}{Hospitals}\\
\hline
$P_{h_{1}}$ & $P_{h_{2}}$ \\
\hline 
$s_1$ & $s_1$ \\
$s_3$ & $s_3$ \\
$s_2$ & $s_2$ 
\end{tabular}
~\qquad~      
    \begin{tabular}{cc}
\multicolumn{2}{c}{Couples}\\
\hline
$P_{(s_1, s_2)}$& $P_{(s_3, s_4)}$ \\
\hline
$(h_1,h_2)$ & $(h_1,u)$ \\ 
            & $(h_2,u)$  \\ 
            & 
\end{tabular} 
\end{table}
Roth (2008) showed that there does not exist a stable matching in this couples market. One can see that there are three individually rational matchings: $\mu_1=\{(s_1,h_1),(s_2,h_2),(s_3,u),(s_4,u)\}$, $\mu_2=\{(s_1,u),(s_2,u),(s_3,h_2),(s_4,u)\}$, and $\mu_3=\{(s_1,u),(s_2,u),(s_3,h_1),(s_4,u)\}$. Notice that hospital $h_1$ and couple $(s_3,s_4)$ are better off at $\mu_3$ than at $\mu_2$, and hence $\mu_3>\mu_2$. Since hospital $h_2$ and couple $(s_3,s_4)$ are better off at $\mu_2$ than at $\mu_1$, it holds that $\mu_2>\mu_1$. Finally, since hospitals $h_1$, $h_{2}$ and couple $(s_1,s_2)$ are better off at $\mu_1$ than at $\mu_3$, we have $\mu_1>\mu_3$. Thus, there is no stable matching. Furthermore, since direct dominance implies indirect dominance, we have that $\mu_1\gg\mu_3\gg\mu_2\gg\mu_1$. Then, any set $V$ containing at least two matchings ($V\supseteq \{\mu_1,\mu_2\}$, $V\supseteq \{\mu_1,\mu_3\}$, and $V\supseteq \{\mu_2,\mu_3\}$) does not satisfy the internal stability condition from Definition \ref{Def:farsightedstableset} and is not a farsighted stable set. Moreover, any singleton set $V$ is not a farsighted stable set since $\mu_3$ does not indirectly dominate $\mu_1$, $\mu_1$ does not indirectly dominate $\mu_2$, and $\mu_2$ does not indirectly dominate $\mu_3$. Then, the external stability condition from Definition \ref{Def:farsightedstableset} would be violated.
\end{aexample}

Thus, there are couples markets where the farsighted stable set does not exist. When agents are farsighted, is there any solution concept that always provides some prediction? Herings, Mauleon and Vannetelbosch (2009, 2010) propose the DEM farsighted stable set. Replacing the internal stability condition in the definition of the farsighted stable set by deterrence of external deviations and minimality, leads to a stability concept, the DEM farsighted stable set, that is always non-empty.

\begin{adefinition}\label{DEMset} 
Let $\langle H, C, P \rangle$ be a couples market. A set of matchings $V\subseteq\mathcal{M}$ is a DEM farsighted stable set if it satisfies:
\begin{enumerate}[label=(\textbf{\roman*})]
\item[(D)] \textit{Deterrence of external deviations}: For any $\mu\in V$, and any $\mu'\notin V$ that can be enforced from $\mu$ by coalition $T\in H \cup C$, there exists $\mu'' \in V$ such that $\mu''\gg\mu'$ and $\mu''\nsucc_T \mu$.
\item[(E)] \textit{External stability}: For every $\mu\notin V$, there exists $\mu'\in V$ such that $\mu'\gg\mu$, 
\item[(M)] \textit{Minimality}: There does not exist $V'\subset V$ such that $V'$ satisfies both (D) and (E).
\end{enumerate}
\end{adefinition}

Condition (D) in Definition \ref{DEMset} requires the deterrence of external deviations. It captures that any deviation to a matching outside $V$, is deterred by the threat of ending in $\mu''$ that indirectly dominates $\mu'$. Moreover $\mu''$ belongs to $V$ which makes $\mu''$ a credible threat. Condition (E) in Definition \ref{DEMset} requires external stability and implies that the matchings within the set are robust to perturbations. Any matching outside of $V$ is indirectly dominated by a matching in the set. Condition (E) implies that if a set of matchings is a DEM farsighted stable set, it is non-empty. Notice that the set $\mathcal{M}$ (trivially) satisfies Conditions (D) and (E) in Definition \ref{DEMset}. This motivates the requirement of a minimality condition, namely Condition (M).

Contrary to farsighted stable sets, DEM farsighted stable sets always exist. Moreover, if $V$ is a farsighted stable set, then $V$ is a DEM farsighted stable set. See Herings, Mauleon and Vannetelbosch (2009, 2010). In Example 4 there is no farsighted stable set. However, the sets $V= \{\mu_1,\mu_2\}$, $V'= \{\mu_1,\mu_3\}$, and $V'' \{\mu_2,\mu_3\}$ are DEM farsighted stable sets. For instance, $V= \{\mu_1,\mu_2\}$ satisfy external stability and deterrence of external deviations because the deviation from $\mu_2$ to $\mu_3$ of coalition $\{h_1,(s_3,s_4)\}$ is deterred by the subsequent deviation of coalition $\{h_1,h_2,(s_1,s_2)\}$ from $\mu_3$ to $\mu_1$ where the couple $(s_3,s_4)$ is worse off than at $\mu_2$. Thus, the DEM farsighted stable set could be used to predict the matchings that are stable when agents are farsighted and the farsighted stable set does not exist. Indeed, every farsighted stable set is a DEM farsighted stable set but the reverse is not true. A characterization of the DEM farsighted stable set is out of the scope of this paper. 

\section{Concluding Remarks}
\label{sec:conc}
We adopt the notion of the farsighted stable set to determine which matchings are stable when agents are farsighted in matching markets with couples. We show that a singleton matching is a farsighted stable set if and only if the matching is stable. Thus, the property of stability of a matching is preserved when agents become farsighted in couples markets. However, other farsighted stable sets with multiple non-stable matchings can exist in couples markets with stable matchings. For markets without stable matching, we provide examples of markets with and without farsighted stable sets. For couples markets where the farsighted stable set does not exist, the DEM farsighted stable set could be used to predict the matchings that are stable when agents are farsighted. 

It has been assumed that each hospital has exactly one position to fill. However, our results can easily be adapted to situations where hospitals have multiple positions if hospitals' preferences are responsive over sets of students.\footnote{A hospital's preferences over groups of students are responsive if, for any two assignments that differ in only one student, it prefers the assignment containing the more preferred student.} Notice that in the proofs we assume that whenever a student leaves its position at a hospital, the hospital would accept an acceptable student to fill the vacant position. When hospitals have multiple positions, we then assume that a hospital always prefer to accept an acceptable student whenever the number of filled positions is lower than the capacity of the hospital. Under this assumption, the results of the paper remain valid for the case of hospitals with multiple positions. 

Recently, the notion of the myopic-farsighted stable set has been proposed to study the interaction between myopic and farsighted agents in one-to-one matching markets. See Herings, Mauleon and Vannetelbosch (2020a, 2020b) and Do\u{g}an and Ehlers (2023). Only when farsighted agents are present on one side of the market, the existence of myopic-farsighted stable sets is guaranteed. Moreover, the introduction of heterogeneous agents or of limited farsighted agents destabilize some stable matchings (the one obtained from the DA algorithm) and stabilize other matchings (the one obtained from the Top Trading Cycle algorithm) in school choice problems and priority-based matching problems. See Atay, Mauleon and Vannetelbosch (2022a, 2022b). An interesting extension would be to study matching with couples when agents are heterogeneous in their degree of farsightedness.

\section*{Acknowledgments}
Ata Atay is a Serra H\'{u}nter Fellow (Professor Lector Serra H\'{u}nter). Ata Atay gratefully acknowledges the support from the Spanish Ministerio de Ciencia e Innovaci\'{o}n research grant PID2020-113110GB-100/AEI/10.13039/501100011033, from the Generalitat de Catalunya research grant 2021SGR00306 and from the University of Barcelona research grant AS017672. Ana Mauleon and Vincent Vannetelbosch are, respectively, Research Director and Senior Research Associate of the National Fund for Scientific Research (FNRS). Financial support from the Fonds de la Recherche Scientifique - FNRS\ research grant T.0143.18 is gratefully acknowledged.

\end{document}